\documentclass{article}
\usepackage{graphicx} 

\usepackage[english]{babel}
\usepackage[utf8]{inputenc}
\usepackage{amsmath,amsfonts,amssymb,amsthm}
\usepackage{tabto}
\usepackage[colorinlistoftodos]{todonotes}
\usepackage{blindtext}
\usepackage{subcaption}
\usepackage{caption}
\usepackage[nodisplayskipstretch]{setspace}
\usepackage{tikz-cd} 
\usepackage{tensor}
\usepackage{mathrsfs}
\usepackage{hyperref}
\usepackage[style=ext-numeric-comp,maxcitenames=2,uniquename=false,backend=biber,uniquelist=false,articlein=false]{biblatex}
\usepackage[usenames,dvipsnames]{xcolor}
\usepackage{enumitem}
\usepackage{eucal}

\DeclareCiteCommand{\citeyear}
    {}
    {\bibhyperref{\printdate}}
    {\multicitedelim}
    {}

\hypersetup{
	colorlinks=true,         
	linkcolor=Black,          
	citecolor=MidnightBlue,
	urlcolor=MidnightBlue            
 }

\newtheorem{thm}{Theorem}

\makeatletter
\renewcommand\section{\@startsection{section}{1}{\z@}{-3.25ex plus -1ex minus -.2ex}{1.5ex plus .2ex}{\normalsize\bf}}
\renewcommand\subsection{\@startsection{subsection}{2}{\z@}{-3.25ex plus -1ex minus -.2ex}{1.5ex plus .2ex}{\normalsize\bf}}
\renewcommand\subsubsection{\@startsection{subsubsection}{3}{\z@}{-3.25ex plus -1ex minus -.2ex}{1.5ex plus .2ex}{\normalsize\bf}}
\makeatother

\title{Two routes to the Hadamard condition}
\author{Eleanor March\footnote{Department of Logic and Philosophy of Science, University of California, Irvine, CA. eleanor.march@uci.edu}}
\date{}

\begin{document}

\maketitle

\begin{abstract}
    In quantum field theory on curved spacetime, the Hadamard condition is often presented as a necessary condition on `physically reasonable', `vacuum-like' states of linear quantum fields, and plays a central role in many theoretical and foundational applications---ranging from proofs of the renormalizability of Wick polynomials to derivations of the Hawking temperature. In this paper, I strengthen, generalise, and bring under technical control two popular motivations for the Hadamard condition---one which appeals to the role of Hadamard states in defining physical quantities such as Wick polynomials and the expectation value of the stress--energy operator on general curved spacetime backgrounds, and one which appeals to the equivalence principle. This clarifies the role and status of the Hadamard condition, including its relationship to the equivalence principle, to well-definedness of physical quantities such as Wick polynomials and the expectation value of the stress--energy operator, and the sense in which Hadamard states are `vacuum-like'. 
\end{abstract}

\section{Introduction}\label{sec:intro}

In quantum field theory on curved spacetime, the Hadamard condition is often presented as a necessary condition on `physically reasonable', `vacuum-like' states of linear quantum fields. Hadamard states permit the definition of important physical quantities such as the expectation value of the stress--energy operator on general curved spacetime backgrounds \parencite{Wald1977,Wald1978,Wald1994,BrunettiFredenhagen2000}, and so play a central role in semi-classical gravity. Many foundational results in the theory of quantum fields on curved spacetimes---ranging from derivations of Hawking radiation and the Hawking temperature \parencite{JanssenVerch2023,KayWald1991}, to quantum energy inequalities \parencite{Fewster2000,FewsterVerch2002,Fewster2017}, to proofs of the renormalizability of Wick polynomials and time-ordered products \parencite{HollandsWald2001,HollandsWald2002}---rely on the Hadamard condition.

For these reasons, providing a robust motivation for the Hadamard condition---one which can underwrite the claim that `physically reasonable', `vacuum-like' states of linear quantum fields are necessarily Hadamard---is a task of considerable theoretical and foundational importance. Standardly, one finds two arguments to this effect in the literature:\footnote{There is also a more recent argument due to \textcite{FewsterVerch2013}, which I do not discuss here.} one (i) which appeals directly to the role of Hadamard states in defining e.g., the expectation value of the stress--energy operator, (renormalized) Wick polynomials and time-ordered products, etc., on general curved spacetime backgrounds (see e.g., \parencite{Fullingetal1978,KayWald1991,JanssenVerch2023,FewsterVerch2013}), and another (ii) which appeals to a version of the equivalence principle, according to which the singularity structure of the two-point function of Hadamard states for linear Klein--Gordon fields should `locally approximate' that of the Minkowski vacuum (see e.g., \parencite{KayWald1991,Amir-HomayoonOttewill,FewsterVerch2013,KhavkineMoretti2015}). 

Recent developments have made it possible to bring under technical control and substantially strengthen and generalise both these arguments; doing so is the aim of the present paper. The microlocal reformulation of the Hadamard condition for Klein--Gordon fields due to \textcite{Radzikowski} paved the way for an important first step towards clarifying (i), which was the observation by \textcite{DuetschFredenhagen2001,DuetschFredenhagen2001-2} (see also e.g., \parencite{FredenhagenRejzner2015,Duetsch}) that, in the context of the perturbative algebraic approach to the quantization of linear Klein--Gordon fields, the Hadamard condition is sufficient for the existence of a unital associative star product on a sufficiently large space of observables of the Klein--Gordon field---essentially, large enough to encompass non-linear polynomials of fields and their derivatives taken pointwise---satisfying some further conditions (concerning, roughly, when that star product defines a deformation quantization of the Klein--Gordon field, and whether the resulting quantum field theory admits states which are `analogues' of the classical vacuum in a precise sense). Recently, \textcite{Fewster2025Hadamard} has extended the microlocal definition of Hadamard states to a more general class of Green-hyperbolic operators (see e.g., \parencite{Baer2015,Khavkine2014}) going well beyond the case of Klein--Gordon theory and other normally-hyperbolic operators. Building on this, the main result of \S\ref{sec:operatorstohadamard} is to generalise the above observation to the class of operators considered in \parencite{Fewster2025Hadamard}, and additionally to prove the converse---i.e., for any such Green-hyperbolic operator, the Hadamard condition is both necessary and sufficient for the existence of a unital associative star product on a sufficiently large space of observables of its associated quantum field, satisfying some further conditions. (To streamline the discussion, this is shown only for the bosonic case in \S\ref{sec:operatorstohadamard}; the extension to fermionic quantization is discussed in Appendix \ref{app:fermions}.) This further clarifies the significance and desirability of the Hadamard condition: if one seeks to quantize perturbatively a Green-hyperbolic operator in such a way that the resulting quantum field theory both has, as observables, non-linear polynomials of fields and their derivatives taken pointwise, and admits any `classical vacuum-like' states at all, then those `classical vacuum-like' states are necessarily Hadamard.

Separately, \textcite{MarchWeatherallFun1} (see also \parencite{March2025BPhil}) have recently provided a mathematically-precise statement of the equivalence principle, which could serve as the basis for a rigorous treatment of (ii). In \S\ref{sec:equivalenceagain}, I undertake such a treatment. The main result of this section is to prove that any state the singularity structure of whose two-point function results from an application of the equivalence principle to that of the Minkowski vacuum of the Klein--Gordon field is Hadamard, at least under some further conditions (those considered in the discussion of (i)). (Actually, the result of \S\ref{sec:equivalenceagain} is slightly more general than this.) This brings, for the first time, technical control to the relationship between the equivalence principle and Hadamard condition along the lines of (ii), as well as opening up several new avenues for further research surrounding this relationship. 

As such, the structure of this paper is as follows. \S\ref{sec:prelims} provides some relevant background on the quantization of Green-hyperbolic operators in the perturbative algebraic setting, and the generalised definition of Hadamard states in this setting given by \textcite{Fewster2025Hadamard}. In \S\ref{sec:operatorstohadamard}, I prove that the bosonic quantization procedure applied to a sufficiently large class of observables leads to Hadamard states and \emph{vice versa} (Appendix \ref{app:fermions} discusses the extension of this result to fermionic quantization). In \S\ref{sec:equivalenceagain}, I introduce the mathematical statement of the equivalence principle in \parencite{March2025BPhil,MarchWeatherallFun1}, and show that applying this to the singularity structure of the Minkowski vacuum of the Klein--Gordon field also leads to Hadamard states (under some further conditions). \S\ref{sec:conclusion} concludes with a discussion.

\section{Quantization of Green-hyperbolic operators in the perturbative algebraic setting}\label{sec:prelims}

\subsection{Preliminaries}

We work with a finite rank complex vector bundle $V\rightarrow E\overset{\pi}{\rightarrow} M$ with dual bundle $E^*$ over a globally-hyperbolic, time-oriented, and oriented Lorentzian $n$-manifold $(M,g)$, where $M$ is assumed paracompact and with finitely many connected components. We write $\mathrm{dvol}_g$ for the density on $M$ induced by $g$. $E$ is assumed to be \emph{Hermitian}, i.e., equipped with a smooth non-degenerate bundle metric $(\,\cdot\,,\,\cdot\,)$, and a base-preserving antiunitary involutive bundle automorphism $C:E\rightarrow E$ (complex conjugation) such that $(C(v),C(w))_x=\overline{(v,w)}{}_x$, and thus a bilinear bundle metric $\langle\,\cdot\,,\,\cdot\,\rangle$ via $\langle v,w\rangle_x:=(C(v),w)_x$. Note that $(\,\cdot\,,\,\cdot\,)$, $\langle\,\cdot\,,\,\cdot\,\rangle$ extend to (respectively sequilinear, bilinear) pairings on sections of $E$ with compactly intersecting supports via integration against $\mathrm{dvol}_g$; the pairing of sections of $E$ and sections of $E^*$ with compactly intersecting supports via integration against $\mathrm{dvol}_g$ is denoted $\langle\langle\,\cdot\,,\,\cdot\,\rangle\rangle$. Note also that $\langle\,\cdot\,,\,\cdot\,\rangle$ induces a vector bundle isomorphism $\sharp:E\rightarrow E^*$ with inverse $\flat$ via $\langle\langle v^\sharp, w\rangle\rangle:=\langle v,w\rangle_x$ for any $v,w\in\pi^{-1}(x)$ (in particular, making $E^*$ into a Hermitian vector bundle).

For any set $S\subset M$, $J^\pm(S)$ denotes its causal future/past. $J^\pm M$ denotes the bundle of (future, past) causal covectors on $M$; $N^\pm M$ denotes the bundle of (future, past) null covectors on $M$; vector bundles overset with a dot denote removal of the zero section. For any finite rank vector bundle $B$ over $(M,g)$, $\Gamma(B)$ denotes its space of smooth sections, $\Gamma_\mathrm{c}(B)$ its space of smooth sections with compact support, $\Gamma_{\mathrm{sc},\pm}(B)$ its space of smooth sections with support in $J^\pm(S)$ for some compact $S\subset M$, and $\Gamma_\mathrm{sc}(B)$ its space of smooth sections with support in $J(S)$ for some compact $S\subset M$. (These are all Fr\'echet topological vector spaces, see Lemma 2.3 of \parencite{Baer2015}.) 

Given any two Hermitian vector bundles $B_1, B_2$ and a linear map $S:\Gamma(B_1)\rightarrow\Gamma(B_2)$, the \emph{formal Hermitian adjoint} of $S$, if it exists, is the unique linear map ${}^\dag\! S:\Gamma(B_2)\rightarrow\Gamma(B_1)$ such that $(f_1,S(f_2))_2=({}^\dag\! S(f_1),f_2)_1$ for any $f_1\in \Gamma(B_1)$, $f_2\in \Gamma(B_2)$ such that $S(f_1)$ and $f_2$ have compactly intersecting supports. The \emph{formal transpose} ${}^t\! S$ of $S$, if it exists, is defined analogously, via $\langle f_1,S(f_2)\rangle_2=\langle{}^t\! S(f_1),f_2\rangle_1$. Likewise, the \emph{formal dual} ${}^\star\! S$ of $S$, if it exists, is defined via $\langle\langle f_1,S(f_2)\rangle\rangle_2=\langle\langle{}^\star\! S(f_1),f_2\rangle\rangle_1$. We also write $\overline{S}=C_2SC_1$.

\subsection{Green-hyperbolic operators}

References for this section include \parencite{Baer2015,Khavkine2014,Fewster2025Hadamard}. We consider the general class of linear differential operators $P:\Gamma(E)\rightarrow\Gamma(E)$.\footnote{For Hermitian vector bundles over oriented Lorentzian manifolds, this includes all linear differential operators which arise as the Euler--Lagrange operator of some Lagrangian.} Any such operator admits a formal Hermitian adjoint. If ${}^\dag\! P=P$, we say that $P$ is \emph{formally Hermitian}. Now suppose that $P$ is formally Hermitian and that there exist linear maps $G^\pm_P:\Gamma_\mathrm{c}(E)\rightarrow\Gamma_{\mathrm{sc},\pm}(E)$ such that for any $f\in\Gamma_{\mathrm{c}}(E)$, (i) $PG^\pm_P(f)=f=G^\pm_PP(f)$, and (ii) $\mathrm{supp}(G^\pm_P(f))\subset J^\pm(\mathrm{supp}(f))$. Then $P$ is a \emph{formally-Hermitian Green-hyperbolic operator} and $G^\mp_P$ are called the \emph{advanced} (respectively \emph{retarded}) \emph{Green functions} for $P$. (Note that the advanced and retarded Green functions of a Green-hyperbolic operator are unique. They are also not, strictly speaking, functions but distributional (generalised) sections.) Their difference $G_P:=G^-_P-G^+_P:\Gamma_\mathrm{c}(E)\rightarrow\Gamma_{\mathrm{sc}}(E)$ is called the \emph{causal Green function}. In what follows, we will also assume that $P$ is real, i.e., $P=\overline{P}:=CPC$. In that case, one has that ${}^tG^\pm_P=G^\mp_P$ and ${}^tG_P=-G_P$, as well as $P G_P(f)=\mathbf{0}=G_PP(f)$ for any $f\in \Gamma_\mathrm{c}(E)$. This allows us to introduce the \emph{causal propagator} for $P$, which is an antisymmetric bilinear form $\Delta_P(\,\cdot\,,\,\cdot\,)$ defined via
\begin{equation*}
    \Delta_P(f,g):=\langle f, G_P(g)\rangle,
\end{equation*}
for all $f,g\in\Gamma_\mathrm{c}(E)$. Note that for any $f,g\in\Gamma_\mathrm{c}(E)$, $\Delta_P(f,P(g))=0=\Delta_P(P(f),g)$ and $\Delta_P(C(f),C(g))=\overline{\Delta_P(f,g)}$. The causal propagator can be regarded as a distributional section of the external tensor product bundle $E^*\boxtimes E^*$.

\subsection{Wavefront sets, $\mathcal{V}^\pm$-decomposability, and Hadamard two-point functions}

To discuss the Hadamard condition in this general context, we need some tools from microlocal analysis. References for this section include \parencite{Hoermander1,Radzikowski,Fewster2025Hadamard}. We begin by introducing the machinery of wavefront sets, due to \textcite{Hoermander1}, which allow us to probe the `directions' in which a distribution is singular.\footnote{Various definitions of the wavefront set are available, see e.g., \parencite{Brouder2014}; the presentation here is adapted from \parencite{Fewster2025Hadamard}.} Let $u$ be a distribution on $M$ and let $(x,k)\in \dot{T}^*M$. Fix a coordinate system $X_\alpha$ in a neighbourhood $U$ of $x$. We say that $(x,k)$ is \emph{non-characteristic} for $u$ iff there is a smooth density $\chi$ with compact support on $U$ and an open cone $V$ in $\dot{T}_x^*M$ such that $\chi(x)\neq \mathbf{0}$, $k\in V$, and 
\begin{equation*}
    \underset{l\in V}{\mathrm{sup}}(1+||l||^N)|u(\chi e_l)|<\infty,
\end{equation*}
for all $N\in\mathbb{N}$ and all $l\in T^*_xM$, where $e_l(p)=e^{iX_\alpha(p)l(\partial_{X_\alpha})}$ and $||\cdot||$ is any norm on $\dot{T}_x^*M$; otherwise $(x,k)$ is \emph{characteristic}.

The characteristic covectors of a distribution track the directions in which $u$ fails to be regular. The \emph{wavefront set} $\mathrm{WF}(u)$ of $u$ is the (closed, conic) subset of $\dot{T}^*M$ of all characteristic covectors for $u$. We write $\mathrm{WF}(u)(x):=\dot{T}^*_xM\cap \mathrm{WF}(u)$. Note that $\mathrm{WF}(u)$ is independent of the choice of coordinate system. The wavefront set extends to the case of distributions acting on sections of complex vector bundles with finite rank, defining:
\begin{equation*}
    \mathrm{WF}(u)=\underset{\varphi\in\Gamma(E^*)}{\bigcup}\mathrm{WF}(f\mapsto u(f\varphi)).
\end{equation*}

Wavefront sets allows us to say when the product of two distributions is well-defined. We focus on the general case of distributions in multiple variables, which is most relevant for the subsequent discussion. Let $V_1\rightarrow E_1\rightarrow M_1$, $V_2\rightarrow E_2\rightarrow M_2$, $V_3\rightarrow E_3\rightarrow M_3$ be complex vector bundles of finite rank. If $u$ is a distributional section of $E_1\boxtimes E_2$ write
\begin{equation*}
    \mathrm{WF}_{M_2}(u):=\{(x_2,k_2)|(x_1,\mathbf{0};x_2,k_2)\in\mathrm{WF}(u)\;\mathrm{for}\;\mathrm{some}\;x_1\in M_1\}.
\end{equation*}
Now let $u_1,u_2$ be distributional sections of $E_1\boxtimes E_2$, $E_2^*\boxtimes E_3$ respectively, and suppose that the fibrewise sum $\mathrm{WF}_{M_2}(u_1)+\mathrm{WF}_{M_2}(u_2)$ does not intersect the zero section of $T^*M_2$ (H\"ormander's criterion). Then their pointwise partial product on $M_2$ is well-defined \parencite{Hoermander1}, and is given via
\begin{equation*}
    (u_1u_2)(f_1,f_3):=\langle\langle u_1(f_1,\cdot\,),u_2(\,\cdot,f_3)\rangle\rangle_2
\end{equation*}
for compactly supported sections $f_1,f_3$ of $E_1^*, E_3^*$ respectively. 

Wavefront sets also put bounds on the singular structure of sums, derivatives, and products of distributions. Three cases are most relevant for us here (\S 8.1 of \parencite{Hoermander1}):
\begin{itemize}
    \item $\mathrm{WF}(u+v)\subseteq \mathrm{WF}(u)\cup\mathrm{WF}(v)$;
    \item If $\nabla$ is any connection on $E$ and $\xi$ is any vector field on $M$, $\mathrm{WF}(\nabla_\xi u)\subset\mathrm{WF}(u)$;
    \item If $u_1,u_2$ are distributional sections of $E_1\boxtimes E_2$, $E_2^*\boxtimes E_3$, and H\"ormander's criterion for the existence of the partial product $u_1u_2$ is met, then $\mathrm{WF}(u_1u_2)(x_1,x_3)\subseteq\{(x_1,k_1;x_3,k_3)\in T^*M\times T^*M: ((x_1,k_1;x_2,k_2)\in\mathrm{WF}(u_1)(x_1,x_2)\;\mathrm{and}\;(x_2,-k_2;x_3,k_3)\in\mathrm{WF}(u_2)(x_2,x_3))\;\mathrm{or}\; (k_1=\mathbf{0}\;\mathrm{and}\;(x_2,\mathbf{0};x_3,k_3)\in\mathrm{WF}(u_2)(x_2,x_3))\;\mathrm{or}\; (k_3=\mathbf{0}\;\mathrm{and}\;(x_1,k_1;x_2,\mathbf{0})\in\mathrm{WF}(u_1)(x_1,x_2))\}$.
\end{itemize}

Following \parencite{Fewster2025Hadamard}, these tools allow us to define an important subclass of real formally-Hermitian Green-hyperbolic operators. Namely, we say that a real formally-Hermitian Green-hyperbolic operator $P$ on a Hermitian vector bundle $E$ is $\mathcal{V}^\pm $\emph{-decomposable} if there exists a conic, relatively closed subset $\mathcal{V}^+$ of $\dot{T}^*M$ (so that $\mathcal{V}^+\cap \mathcal{V}^-=\varnothing$, where $\mathcal{V}^-:=-\mathcal{V}^+$) such that 
\begin{equation*}
    \mathrm{WF}(\Delta_P)\subset(\mathcal{V}^+\times \mathcal{V}^-)\cup(\mathcal{V}^-\times \mathcal{V}^+).
\end{equation*}
Roughly, one can think of $\mathcal{V}^+$-decomposability as saying that the singular structure of $\Delta_P$ can be split into a well-defined `future-past' and `past-future' component, but where these are not required to be defined with respect to the metric lightcones. In particular, for a $\mathcal{V}^+$-decomposable real formally-Hermitian Green-hyperbolic operator, this allows us to define a causal (i.e., lightcone-like) structure associated to its causal propagator. That is, we say that a covector $k$ at $x\in M$ is $\mathcal{V}$\emph{-causal} if it lies in the (fibrewise) convex hull $\mathcal{U}^\pm$ of $\mathcal{V}^\pm$ and $\mathcal{V}$\emph{-spacelike} otherwise. Following \parencite{Fewster2025Hadamard}, we can now introduce the Hadamard condition in this generalized context. Let $E$ be a Hermitian vector bundle and $P$ a $\mathcal{V}^\pm$-decomposable real formally-Hermitian Green-hyperbolic operator thereon. We say that a distributional section $\Delta_H$ of $E^*\boxtimes E^*$ is a (bosonic) $\mathcal{V}^+$\emph{-Hadamard} two-point function for $P$ iff
\begin{itemize}
    \item (Antisymmetric part) $\Delta_H-{}^t\Delta_H=\Delta_H-\overline{\Delta}_H=2i\Delta_P$;
    \item (Wavefront set spectrum condition (WFSSC)) Its wavefront set is a subset of one causal half of that of the causal propagator for $P$, i.e.,
    \begin{equation*}
        \mathrm{WF}(\Delta_H)\subset\mathcal{V}^+\times\mathcal{V}^-;
    \end{equation*}
    \item (Bi-solution) It is a distributional bi-solution of $P$, in the sense that for any $f,g\in\Gamma_c(E)$, $\Delta_H(P(f),g)=0=\Delta_H(f,P(g))$; and
    \item (Positivity) It is positive semi-definite, i.e.\ for any $f\in \Gamma_\mathrm{c}(E)$, $\Delta_H(\overline{f},f)\geq 0$.
\end{itemize}

We defer discussion of the relationship between a (bosonic) Hadamard two-point function and a Hadamard state to \S\ref{sec:operators}.

\subsection{Observables}\label{sec:observables}
We now discuss how a real formally-Hermitian Green-hyperbolic field theory can be used to construct a bosonic quantum field theory (for the extension to the fermionic case, see Appendix \ref{app:fermions}). We begin by defining an appropriate space of observables. References for this section include \parencite{Duetsch,DuetschFredenhagen2001,DuetschFredenhagen2001-2,Rejzner2022,FredenhagenRejzner2015,Fewster2025Hadamard,Khavkine2014,Beninietal2015}. Given a Hermitian vector bundle $E$, an \emph{observable} on $E$ in general is a continuous map $A:\Gamma(E)\rightarrow\mathbb{C}$. In what follows, we will focus on observables which are \emph{polynomial}, in the sense of being formal power series in sections of $E$.\footnote{Recall that that a power series is formal means that no requirement is made that it be convergent.} In some detail, this means that for any polynomial observable $A$ there exist compactly-supported distributional sections $\alpha^{(i)}$ of $\boxtimes^{i}_\mathrm{sym}E$, $i=0,1,2,...$, such that $A(f)=\sum_{i=0}\langle\alpha^{(i)},\otimes^{i}f\rangle$ for all $f\in\Gamma(E)$, where the subscript $\mathrm{sym}$ denotes symmetrization. The space of polynomial observables on $E$ is denoted $\mathcal{P}(E)$. For the space of polynomial observables which are formal power series in some parameter, $\hbar$, we write $\mathcal{P}(E)[[\hbar]]$.

Two subspaces of $\mathcal{P}(E)$ will be of interest to us. First, the \emph{regular polynomial observables}, denoted $\mathcal{P}_\mathrm{reg}(E)$, for which the coefficients $\alpha^{(i)}$ are all regular distributions. This is rather restrictive. Second, for any conic, relatively closed $\mathcal{V}^+\subset \dot{T}^*M$ the $\mathcal{V}$\emph{-microcausal polynomial observables}, denoted $\mathcal{P}_{\mathcal{V}\mathrm{mc}}(E)$,\footnote{The definition of a $\mathcal{V}$-microcausal observable generalises the definition of a `microcausal observable' in \parencite{FredenhagenRejzner2015,Rejzner2022}, as well as the definition of a `field' in \parencite{Duetsch}.} for which the wavefront set of each $\alpha^{(i)}$ is disjoint from those points where all $i$ wavevectors are in the closed future cone or closed past cone of $\mathcal{V}$-causal covectors, i.e.,
\begin{equation*}
    \mathrm{WF}(\alpha^{(i)})\cap (\times^i\mathcal{U}^+\cup \times^i\mathcal{U}^-)=\varnothing,
\end{equation*}
where as before $\mathcal{U}^\pm$ is the (fibrewise) convex hull of $\mathcal{V}^\pm$. Note that for any $\mathcal{V}$-microcausal polynomial observable, $\alpha^{(1)}$ is either regular or its wavefront set is $\mathcal{V}$-spacelike. The most obvious example of a $\mathcal{V}$-microcausal polynomial observable which is not a regular polynomial observable is the squared field amplitude, for which $\alpha^{(2)}(x,y)=\delta(x,y)\langle\,\cdot\,,\,\cdot\,\rangle$. More generally, $\mathcal{P}_{\mathcal{V}\mathrm{mc}}(E)$ (but not $\mathcal{P}_\mathrm{reg}(E)$) is rich enough to encompass polynomials of fields and their derivatives which are taken pointwise.

We now consider dynamics. Let $P$ be a real formally-Hermitian Green-hyperbolic operator on $E$. For $A_1, A_2\in \mathcal{P}(E)$, define $A_1\sim_PA_2$ iff $A_1(f)=A_2(f)$ for all $f\in\mathrm{ker}(P)$. The space of \emph{on-shell} (regular, $\mathcal{V}$-microcausal) polynomial observables for $P$, denoted $\mathcal{P}(E,P)$ (respectively $\mathcal{P}_\mathrm{reg}(E,P)$, $\mathcal{P}_{\mathcal{V}\mathrm{mc}}(E,P)$) is the quotient $\mathcal{P}(E)/\!\!\sim_P$ (respectively $\mathcal{P}_\mathrm{reg}(E)/\!\!\sim_P$, $\mathcal{P}_{\mathcal{V}\mathrm{mc}}(E)/\!\!\sim_P$). It is easy to see (using that $P$ is formally Hermitian) that $A_1\sim_PA_2$ iff $A=A_1-A_2$ is such that $\alpha^{(i)}\in\mathrm{im}(P\otimes(\otimes^{i-1}\mathrm{id}))$ for all $i$. Thus, denoting as $\mathcal{I}^P(E)$ the space of all observables for which $\alpha^{(i)}\in\mathrm{im}(P\otimes(\otimes^{i-1}\mathrm{id}))$ for all $i$, we see that $\mathcal{P}(E,P)\simeq\mathcal{P}(E)/\mathcal{I}^P(E)$ (\emph{mutatis mutandis} for regular, microcausal polynomial observables). 

Any $\mathcal{V}$-decomposable real formally-Hermitian Green-hyperbolic operator $P$ on $E$ defines a \emph{Poisson--Peierls} bracket $[\,\cdot\,,\,\cdot\,]_P$ (see \parencite{Peierls1952}) on $\mathcal{P}_{\mathcal{V}\mathrm{mc}}(E)$ via
\begin{equation}
    [A_1,A_2]_P:=\mathrm{prod}\circ \left\langle \Delta_P, \delta\otimes\delta\right\rangle A_1\otimes A_2,\label{eq:Peierls}
\end{equation}
where $\mathrm{prod}$ is the multiplication operator, i.e., $\mathrm{prod}\circ(A_1\otimes A_2)(f)=A_1(f)A_2(f)$, and $\delta$ denotes the functional derivative, so that $(\delta A(f))(g)=\sum_{i=1}\langle\alpha^{(i)},g\otimes (\otimes^{i-1}f)\rangle$. It is straightforward to check using H\"ormander's criterion that $[\,\cdot\,,\,\cdot\,]_P$ is indeed well-defined on $\mathcal{P}_{\mathcal{V}\mathrm{mc}}(E)$.\footnote{The argument is a straightforward modification of that given explicitly in e.g., \S4 of \parencite{FredenhagenRejzner2015}, where $\mathcal{V}^\pm$ is no longer required to be $J^\pm$.} Note that $[\,\cdot\,,\,\cdot\,]_P$ descends to a symplectic form on $\mathcal{P}_{\mathcal{V}\mathrm{mc}}(E,P)$ \parencite{Fewster2025Hadamard,Duetsch,Khavkine2014}. 

\subsection{Operator products}\label{sec:operators}

We now discuss the construction of an appropriate operator product. References for this section include \parencite{Duetsch,DuetschFredenhagen2001,DuetschFredenhagen2001-2,Rejzner2022,FredenhagenRejzner2015}. We begin with a general discussion of the situation for $\mathcal{P}_\mathrm{reg}(E)$.
Any distributional section $\Pi$ of $E^*\boxtimes E^*$ induces a unital associative algebra product on $\mathcal{P}_\mathrm{reg}(E)[[\hbar]]$, denoted $\star_\Pi$, via
\begin{equation}
    A_1\star_\Pi A_2:=\mathrm{prod}\circ\mathrm{exp} \left(\hbar \left\langle \Pi, \delta\otimes\delta\right\rangle\right)A_1\otimes A_2\label{eq:starprod}
\end{equation}
where $\mathrm{exp}$ is in the sense of formal power series. In bosonic quantization, we require that this product satisfy the following conditions:
\begin{enumerate}
    \item $\star_\Pi$ is a star product, i.e., $\overline{(A_1\star_\Pi A_2)}=\overline{A}_2\star_\Pi\overline{A}_1$;
    \item $\star_\Pi$ descends to $\mathcal{P}_\mathrm{reg}(E,P)[[\hbar]]$, i.e., $\mathcal{I}^P_\mathrm{reg}$ is a two-sided ideal of $\star_\Pi$; and
    \item $\star_\Pi$ is a deformation quantization of $[\,\cdot\,,\,\cdot\,]_P$, i.e., as $\hbar\rightarrow 0$, $A_1\star_{\Pi}A_2\rightarrow \mathrm{prod}(A_1\otimes A_2)$ and $(A_1\star_{\Pi}A_2-A_2\star_{\Pi}A_1)/i\hbar\rightarrow [A_1,A_2]_{P}.$
\end{enumerate}

In preparation for \S\ref{sec:operatorstohadamard}, we note that (1) is the case iff $\Pi$ is formally Hermitian; (2) is the case iff $\Pi(P(f),g)=0=\Pi(P(g),f)$ for all $f,g\in\Gamma_c(E)$; and (3) is the case iff $\Pi-{}^t\Pi=i\Delta_P$.\footnote{The first and third of these claims follow straightforwardly from \eqref{eq:Peierls} and \eqref{eq:starprod} by definition of the formal ajoint/formal transpose respectively. The second claim follows from the series expansion of \eqref{eq:starprod}, noting that, if $A_1\in \mathcal{I}^P_\mathrm{reg}(E)$, then each $i$th-order term in this expansion has coefficients which either lie in $\mathrm{im}(P\otimes(\otimes^{i-1}\mathrm{id}))$ automatically (since those of $A_1$ do), or which result from evaluating one argument of $\Pi$ against the first argument of some $\alpha_1^(i)\in \mathrm{im}(P\otimes(\otimes^{i-1}\mathrm{id}))$ against $\Pi$ in its first argument. If $\Pi(P(f),g)=\mathbf{0}=\Pi(P(g),f)$ for all $f,g\in\Gamma(E)$ then these terms vanish. Conversely, If $A_1$ is linear and $A_2$ is quadratic, we see in particular that $\Pi(P(f),g)g\in\mathrm{im}(P)$ for any $f,g\in\Gamma_c(E)$. This implies that $\Pi(P(f),g)g=\mathbf{0}$ for any $f,g\in\Gamma_c(E)$. The same reasoning applied to $A_2$ establishes the result for the second argument of $\Pi$.} The most obvious choice of $\Pi$ satisfying these conditions is simply to take $\Pi=\frac{i}{2}\Delta_P$. Evidently, however, this choice will not work if we wish to extend this treatment to larger spaces of observables, such as $\mathcal{P}_{\mathcal{V}\mathrm{mc}}(E)$.

Finally, note that any bosonic $\mathcal{V}^+$-Hadamard two-point function $\Delta_H$ defines a (quasifree) Hadamard \emph{state} on the algebra $(\mathcal{P}_\mathrm{reg}(E,P)[[\hbar]], \star_{\frac{i}{2}\Delta_P})$, with $\Delta_H$ as the two-point function of that state \parencite{Fewster2025Hadamard,Duetsch,FredenhagenRejzner2015}, where we recall that a state on a unital star algebra $(\mathcal{A}, \star)$ is a linear map $\omega:\mathcal{A}\rightarrow \mathbb{C}$ taking the multiplicative unit to $1$ and such that $\omega(\overline{A}\star A)\geq 0$ for all $A\in\mathcal{A}$. For the relationship between this definition and the definition in terms of the `Hadamard form' in e.g., \parencite{KayWald1991,Fullingetal1978,Wald1994}, see \parencite{Radzikowski,Fewster2025Hadamard}.

\section{From operator products to the Hadamard condition}\label{sec:operatorstohadamard}

We now discuss in more detail extensions of the quantization procedure in \S\ref{sec:operators} to $\mathcal{P}_{\mathcal{V}\mathrm{mc}}(E)$. We have the following result:
\begin{thm}\label{thm:operatorstohadamard}
    Let $E$ be a Hermitian vector bundle, $\Pi$ a distributional section of $E^*\boxtimes E^*$, and $P$ a $\mathcal{V}^\pm$-decomposable real formally-Hermitian Green-hyperbolic operator on $E$. The following conditions are equivalent:
    \begin{itemize}
        \item $\Pi$ is a (bosonic) $\mathcal{V}^+$-Hadamard or $\mathcal{V}^-$-Hadamard two-point function for $P$.
        \item The unital associative algebra product $\star_{\Pi}$ induced by $\Pi$ is well-defined on $\mathcal{P}_{\mathcal{V}\mathrm{mc}}(E)[[\hbar]]$, in the sense that H\"ormander's criterion is satisfied, and has the following properties:
        \begin{enumerate}
            \item $\star_\Pi$ is a star product on $\mathcal{P}_{\mathcal{V}mc}(E)[[\hbar]]$;
            \item $\star_\Pi$ descends to $\mathcal{P}_{\mathcal{V}mc}(E,P)[[\hbar]]$;
            \item $\star_\Pi$ is a deformation quantization of $[\,\cdot\,,\,\cdot\,]_P$; and
            \item The map $\langle\,\cdot\,\rangle_\mathbf{0}:\mathcal{P}_{\mathcal{V}mc}(E)[[\hbar]]\rightarrow \mathbb{C}[[\hbar]]$, $\langle A\rangle_\mathbf{0}:=A(\mathbf{0})$ is a state on the algebra $(\mathcal{P}_{\mathcal{V}mc}(E)[[\hbar]],\star_\Pi)$.
        \end{enumerate}
    \end{itemize}
\end{thm}
\begin{proof}
    We know from $\S\ref{sec:operators}$ that (1)--(3) are the case iff $\Pi-{}^t\Pi=\Pi-\overline{\Pi}=2i\Delta_P$ and $\Pi(P(f),g)=0=\Pi(P(g),f)$ for all $f,g\in\Gamma_c(E)$, i.e., (1)--(3) are equivalent to (antisymmetric part) and (bisolution). It follows from (antisymmetric part) that     \begin{equation}
        \mathrm{WF}(\Pi)=\mathrm{WF}(2i\Delta_P+{}^t\Pi)=(\mathrm{WF}(\Delta_P)\cup\mathrm{WF}({}^t\Pi))\cap\mathrm{WF}(\Pi).\label{eq:antisymmetric} 
    \end{equation}
    Moreover, it is straightforward to see that, assuming $\star_\Pi$ is well-defined on $\mathcal{P}_{\mathcal{V}\mathrm{mc}}(E)[[\hbar]]$, $\langle\,\cdot\,\rangle_\mathbf{0}$ is a state on $\mathcal{P}_{\mathcal{V}\mathrm{mc}}(E,P)[[\hbar]]$ iff (positivity) holds. It is clear that $\langle\,\cdot\,\rangle_\mathbf{0}$ is linear and takes the multiplicative unit to $1$. For positivity, we have explicitly
    \begin{equation*}
        \langle \overline{A}\star_\Pi A\rangle_\mathbf{0}=\overline{\alpha}{}^{(0)}\alpha^{(0)}+\sum_{i=1}\frac{\hbar^i}{i!}\langle\langle \overline{\alpha}{}^{(i)},\Pi^i\alpha^{(i)}\rangle\rangle,
    \end{equation*}
    from which it is clear that (positivity) implies $\langle \overline{A}\star_\Pi A\rangle_\mathbf{0}\geq 0$ and \emph{vice versa}.
    
    Thus, using \eqref{eq:antisymmetric}, it suffices to show that $\star_\Pi$ is well-defined on $\mathcal{P}_{\mathcal{V}\mathrm{mc}}(E)[[\hbar]]$, in the sense that H\"ormander's criterion is satisfied, iff the following conditions are met:
    \begin{itemize}
        \item[(a)] $\mathrm{WF}(\Pi)$ contains no $\mathcal{V}^\pm$-spacelike covectors (in either argument);
        \item[(b)] $\mathrm{WF}(\Pi)$ contains no covectors of the form $(x,\mathbf{0};y,k_y)$ or $(x,k_x;y,\mathbf{0})$; and
        \item[(c)] For any $(x,k_x;y,k_y)\in\mathrm{WF}(\Pi)$, there is no $(z,k_z;u,k_u)\in\mathrm{WF}(\Pi)$ such that $k_x$ and $k_z$ are $\mathcal{V}^\pm$ oppositely oriented and $k_y$ and $k_u$ are $\mathcal{V}^\pm$ oppositely oriented.
    \end{itemize}
    It is clear that (a)--(c) are satisfied if $\Pi$ satisfies WFSSC for $\mathcal{V}^+$ or $\mathcal{V}^-$. For the converse, note that (a)--(c) imply $\mathrm{WF}(\Pi)\cap\mathrm{WF}({}^t\Pi)=\varnothing$, so that, by \eqref{eq:antisymmetric}, $\mathrm{WF}(\Pi)\subseteq\mathrm{WF}(\Delta_P)$. Thus by (c), for any $(x;y)\in M\times M$, $\mathrm{WF}(\Pi)(x;y)\subset\mathcal{V}^+_x\times\mathcal{V}^-_y$ or $\mathrm{WF}(\Pi)(x;y)\subset\mathcal{V}^-_x\times\mathcal{V}^+_y$, and also for any $(x;y), (z;u)\in M\times M$, $\mathrm{WF}(\Pi)(x;y)\subset\mathcal{V}^+_x\times\mathcal{V}^-_y$ iff $\mathrm{WF}(\Pi)(z;u)\subset\mathcal{V}^+_z\times\mathcal{V}^-_u$. Thus $\mathrm{WF}(\Pi)\subset \mathcal{V}^+\times\mathcal{V}^-$ or $\mathrm{WF}(\Pi)\subset \mathcal{V}^-\times\mathcal{V}^+$ giving us WFSSC. We now take (a)--(c) in turn. 

    Consider the series expansion of \eqref{eq:starprod}. For (a), consider the term in this expansion of the form $\mathrm{prod}\circ\langle \Pi(x,y),\alpha^{(1)}_1(x) \otimes\alpha^{(1)}_2(y)\rangle$. For $A_1, A_2\in \mathcal{P}_{\mathcal{V}\mathrm{mc}}(E)[[\hbar]]$, we know that $\mathrm{WF}(\alpha^{(1)}_1)$, $\mathrm{WF}(\alpha^{(1)}_2)$ are $\mathcal{V}$-spacelike or empty. Suppose for contradiction that there exists some $(x,k_x;y,k_y)\in\mathrm{WF}(\Pi)$ with $k_x$ or $k_y$ $\mathcal{V}$-spacelike. By Theorem 8.1.4 of \parencite{Hoermander1}, we know that $A_1$, $A_2$ can be chosen so that $\mathrm{WF}(\alpha^{(1)}_1)$, $\mathrm{WF}(\alpha^{(1)}_2)$ contain $(x,-k_x)$, $(y,-k_y)$ respectively. In this case, H\"ormander's criterion for the existence of the product $\mathrm{prod}\circ\langle \Pi,\alpha^{(1)}_1\otimes \alpha^{(1)}_2\rangle$ is not met.

    We now consider the term in the series expansion of \eqref{eq:starprod} of the form $\mathrm{prod}\circ\langle \Pi(x,y)\otimes\Pi(z,u),\alpha^{(2)}_1(x,z)\otimes \alpha^{(2)}_2(y,u)\rangle$. We evaluate each partial product in steps. For (b), suppose for contradiction that there exists some $(x,k_x;y,k_y)\in\mathrm{WF}(\Pi)$ with $k_x=\mathbf{0}$ or $k_y=\mathbf{0}$. Then by Theorem 8.1.4 of \parencite{Hoermander1}, we may choose $A_1$, $A_2$ such that $(x,-k_x;z,\mathbf{0})\in\mathrm{WF}(\alpha_1^{(2)})$ for some $u\in M$, $(y,-k_y;u,\mathbf{0})\in\mathrm{WF}(\alpha_2^{(2)})$ for some $z\in M$ respectively. In this case, H\"ormander's criterion for the existence of the partial product $\langle \Pi(x,y),\alpha^{(2)}_1(x,z)\alpha^{(2)}_2(y,u) \rangle$ is not met. 

    Now consider $\langle \Pi(x,y),\alpha^{(2)}_1(x,z)\alpha^{(2)}_2(y,u)\rangle$. By (b), this has wavefront set constrained by $\mathrm{WF}(\langle \Pi(x,y),\alpha^{(2)}_1(x,z)\alpha^{(2)}_2(y,u)\rangle)\subseteq\{ (z,k_z;u,k_u): k_z,k_u\neq\mathbf{0}\;\mathrm{and}\;((x,k_x;y,k_y)\in \mathrm{WF}(\Delta_H(x,y))\;\mathrm{and}\; (x,-k_x;z,k_z)\in \mathrm{WF}(\alpha^{(2)}_1)\;\mathrm{and}\; (y,-k_y;u,k_u)\in \mathrm{WF}(\alpha^{(2)}_2))\}$, where we have made use of the definition of a $\mathcal{V}$-microcausal polynomial observable. Thus we see by definition of a $\mathcal{V}$-microcausal polynomial observable that for any such $(z,k_z;u,k_u)$ contained in this set with $k_z,k_u$ $\mathcal{V}$-causal (which is the only case that matters, by (a)), $k_z$ is $\mathcal{V}$-co-oriented with $k_u$ iff $k_x$ is $\mathcal{V}$-co-oriented with $k_y$. So again by Theorem 8.1.4 of \parencite{Hoermander1}, we see that H\"ormander's criterion for the existence of $\mathrm{prod}\circ\langle \Pi(x,y)\otimes\Pi(z,u),\alpha^{(2)}_1(x,z)\otimes \alpha^{(2)}_2(y,u)\rangle$ is met only if for any $(x,k_x;y,k_y)\in\mathrm{WF}(\Pi)$, there is no $(z,k_z;u,k_u)\in\mathrm{WF}(\Pi)$ such that $k_x$ is $\mathcal{V}$-oppositely-oriented with $k_z$ and $k_y$ is $\mathcal{V}$-oppositely-oriented with $k_u$, which gives us (c). Moreover, the same arguments used for (a)--(c) clearly suffice to establish that H\"ormander's criterion is met for every term in the series expansion of \eqref{eq:starprod} whenever $A_1, A_2\in \mathcal{P}_{\mathcal{V}\mathrm{mc}}(E)[[\hbar]]$. 
\end{proof}

The lesson of this result is as follows. Suppose that one seeks to define a bosonic (again, see Appendix \ref{app:fermions} for the extension to the fermionic case) perturbative quantum field theory on $\mathcal{P}_{\mathcal{V}\mathrm{mc}}(E)$, i.e., a perturbative quantum field theory which has, as observables, non-linear polynomials of fields and their derivatives, which is a deformation quantization of the classical field theory $P$, and has a well-defined restriction on-shell. Suppose moreover that one wants this quantum field theory to admit states which are `analogues' of the classical vacuum, in the sense that they consist in evaluating all observables against the classical vacuum field configuration. Then those `vacuum-like' states are necessarily Hadamard (where we have made use of the fact that the state $\langle\,\cdot\,\rangle_\mathbf{0}$ considered in theorem \ref{thm:operatorstohadamard} has $\Pi$ as its two-point function). In other words, Theorem \ref{thm:operatorstohadamard} provides a precise sense in which `vacuum-like' (\emph{per} condition (4)) states of `physically reasonable' (\emph{per} conditions (1)-(3)) and `sufficiently rich' (\emph{per} being defined on $\mathcal{P}_{\mathcal{V}\mathrm{mc}}(E)[[\hbar]]$) quantum field theories are necessarily Hadamard. This, I take it, is a substantial strengthening and clarification of motivation (i) discussed in the introduction. 

\section{From the equivalence principle to the Hadamard condition}\label{sec:equivalenceagain}

I now turn to the relationship between the equivalence principle and the Hadamard condition. This requires us to work at a slightly greater level of generality than we have been thus far. Let $\mathsf{Man}$ be the category whose objects are smooth $n$-manifolds, and whose arrows are smooth embeddings. Let $\mathsf{FB}$ be the category of fibre bundles whose arrows are smooth bundle morphisms. A \emph{natural bundle} over $\mathsf{Man}$ is a functor $\mathscr{B}:\mathsf{Man}\rightarrow\mathsf{FB}$ such that for any object $M$ in $\mathsf{Man}$, $\mathscr{B}M$ is a bundle with base space $M$, and for any morphism $\varphi:M\rightarrow M'$ in $\mathsf{Man}$, $\mathscr{B}\varphi$ is a smooth bundle morphism covering $\varphi$ which is a fibrewise diffeomorphism. Observe that any natural bundle over $\mathsf{Man}$ is also a natural bundle over $\mathsf{Loc}$, the category whose objects are smooth, globally hyperbolic, time-oriented Lorentzian manifolds, and whose arrows are smooth, causal, orientation-preserving, isometric embeddings (the converse is not true). 

Now suppose we have a functor which assigns, to every object $(M,g)$ in $\mathsf{Loc}$ a subset $\mathscr{S}(M,g)$ of a natural bundle $\mathscr{B}M$ (over $\mathsf{Man}$). Then the equivalence principle says that $\mathscr{S}$ depends only on the metric and orientation to first order, in the sense that whenever two objects $(M,g)$, $(M',g')$ in $\mathsf{Loc}$ agree on the metric and orientation to first order at a point (under some local diffeomorphism), then they also agree on $\mathscr{S}$ at that point (under the same local diffeomorphism). To make this more intuitive, note that this is equivalent to saying that for any object $(M,g)$ in $\mathsf{Loc}$, and any $x\in M$, $\mathscr{S}(M,g)$ agrees (at $x$) with the subset of $\mathscr{B}$ which $\mathscr{S}$ assigns to any flat metric $g'$ defined on a neighbourhood $U$ of $x$ and agreeing with $g$ to first order at $x$ (and hence `locally approximates' $\mathscr{S}(U,g')$ in sufficiently small neighbourhoods of $x$) \parencite{March2025BPhil,MarchWeatherallFun1}. (Note that we cannot say this using \emph{only} the category $\mathsf{Loc}$. Note also that this implies that $\mathscr{S}$ must be functorial over $\mathsf{Loc}$.) Letting $\mathscr{S}$ assign to each $(M,g)$ a closed embedded submanifold of a jet bundle, i.e., a partial differential equation, at each $(M,g)$, one recovers the more standard idea of `equivalence principle-satisfying' matter dynamics `locally approximating' those in Minkowski spacetime. 

Following \parencite{March2025BPhil,MarchWeatherallFun1}, to make this precise we proceed as follows (presented here in a slightly generalized form). Let $\mathscr{B}:\mathsf{Man}\rightarrow\mathsf{FB}$, let $\mathscr{F}:\mathsf{FB}\rightarrow\mathsf{Set}$ be the forgetful functor which takes each fibre bundle to its underlying set, and let $\mathscr{L}:\mathsf{Loc}\rightarrow\mathsf{Man}$ be the forgetful functor which takes each globally hyperbolic, time-oriented Lorentzian manifold to its underlying smooth manifold. Let $\mathscr{S}:\mathsf{Loc}\rightarrow\mathsf{Set}$ be a functor, and $s:\mathscr{S}\rightarrow\mathscr{F}\circ\mathscr{B}\circ\mathscr{L}$ be a monic natural transformation. We say that $(\mathscr{S},s)$ \emph{satisfies the equivalence principle} iff, for any objects $(M,g)$, $(M',g')$ in $\mathsf{Loc}$, any $x\in M$, $x'\in M'$ any $U\subset M$, $U'\subset M'$ containing $x$, $x'$ respectively, and any diffeomorphism $\varphi:U\rightarrow U'$ such that $\varphi(x)=x'$, $(\varphi_*g)(x')=g'(x')$, $(\varphi_*\nabla)(x')=\nabla'(x')$ (where $\nabla$, $\nabla'$ are the Levi-Civita connections for $g$, $g'$ respectively), and which preserves the time orientation at $x$, $\mathscr{B}\varphi((s_{(M,g)}\circ\mathscr{S}(M,g))\cap\pi^{-1}_{\mathscr{B}M}(x))=(s_{(M',g')}\circ\mathscr{S}(M',g'))\cap\pi^{-1}_{\mathscr{B}M'}(x')$.

We now seek to apply this to the Hadamard condition. Recall the motivating idea behind the equivalence principle argument for the Hadamard condition (ii) in \S\ref{sec:intro}: the singularity structure of a Hadamard two-point function for the Klein--Gordon field ought to `look' like that of the Minkowski vacuum locally. We cannot apply the above definition of the equivalence principle to a Hadamard two-point function directly (since Hadamard two-point functions are distributional, and in any case, cannot in general be functorial over $\mathsf{Loc}$ \parencite{FewsterVerch2012localitycovariance,Fewster2018states}). But we \emph{can} apply it to the wavefront set of a distribution on each object $(M,g)$ of $\mathsf{Loc}$ (providing those wavefront sets define a natural bundle over $\mathsf{Loc}$). Now, this will not quite work for making sense of the idea of a wavefront set of a Hadamard two-point function satisfying the equivalence principle, since formally, the definition of a natural bundle only applies to bundles defined over $M$ (whilst $\mathrm{WF}(\Delta_H)$ is defined over $M\times M$), but the fix is straightforward. We start by amending our definition of a natural bundle to consider bundles defined over multiple copies of $M$. That is, a $k$-natural bundle (over $\mathsf{Man}$) is a functor $\mathscr{B}:\mathsf{Man}\rightarrow\mathsf{FB}$ such that for any object $M$ in $\mathsf{Man}$, $\mathscr{B}M$ is a bundle with base space $\times^kM$, and (ii) for any morphism $\varphi:M\rightarrow M'$ in $\mathsf{Man}$, $\mathscr{B}\varphi$ is a smooth bundle morphism covering $\times^k\varphi$ which is a fibrewise diffeomorphism. Then given any $k$-natural bundle $\mathscr{B}:\mathsf{Man}\rightarrow\mathsf{FB}$, functor $\mathscr{S}:\mathsf{Loc}\rightarrow\mathsf{Set}$, and monic natural transformation $s:\mathscr{S}\rightarrow\mathscr{F}\circ\mathscr{B}\circ\mathscr{L}$ we can say that $(\mathscr{S},s)$ satisfies the equivalence principle iff for any objects $(M,g)$, $(M',g')$ in $\mathsf{Loc}$, any $x_1,...,x_k\in M$, $x'_1,...,x'_k\in M'$ any $U\subset M$, $U'\subset M'$ containing $x_1,...,x_k$, $x'_1,...,x'_k$ respectively, and any diffeomorphism $\varphi:U\rightarrow U'$ such that $\varphi(x_i)=x'_i$, $(\varphi_*g)(x'_i)=g'(x'_i)$, $(\varphi_*\nabla)(x'_i)=\nabla'(x'_i)$ for all $i=1,...,k$, and which preserves the time orientation at $x_1,...,x_k$, $\mathscr{B}\varphi((s_{(M,g)}\circ\mathscr{S}(M,g))\cap\pi^{-1}_{\mathscr{B}M}(x_1,...,x_k))=(s_{(M',g')}\circ\mathscr{S}(M',g'))\cap\pi^{-1}_{\mathscr{B}M'}(x'_1,...,x'_k)$.\footnote{One could consider a weaker extension of the equivalence principle to the $k$-natural bundles context, which would apply only to the case where $x_1=...=x_k$ and $x'_1=...=x'_k$, i.e., to $k$-tuples of points on the diagonals in $\times^kM$, $\times^kM'$. (Note that this would not be sufficient for theorem \ref{thm:SEPforhadamard} stated below.) However, there is reason to think that the statement of the equivalence principle presented in the main text better captures the idea behind the statement of the equivalence principle in \parencite{March2025BPhil,MarchWeatherallFun1}. That is, the idea of the equivalence principle is that points which agree on the metric to first order should agree on the fibre of $\mathscr{S}$ at those points, so that $\mathscr{S}$ depends only on the metric and orientation to first order. To capture the same notion of dependence only on the metric and orientation to first order in the $k$-natural bundles context, one should consider points in $\times^kM$ regardless of whether they lie on the diagonal or not. I am grateful to Jeff Barrett for helpful comments and discussion on this point.}

We can now make precise sense of the idea that the Hadamard condition, at least in some special cases, is entailed by the equivalence principle applied to the singularity structure of the Minkowski vacuum. The proof relies on the following result due to \parencite{FletcherWeatherallpart1} (which I have translated into the notation used in this paper):
\begin{thm}[Fletcher and Weatherall, 2023]\label{thm:FletcherWeatherall}
    Let $(M,g)$, $(M',g')$ be two Lorentzian $n$-manifolds, let $x\in M$, $x'\in M'$ and let $\gamma':I'\rightarrow M'$ be a smooth embedded curve whose image contains $x'$. Then there exist a smooth embedded curve $\gamma:I\rightarrow M$ whose image contains $x$, neighbourhoods $U$, $U'$ of $x$, $x'$ respectively, and a diffeomorphism $\varphi:U'\rightarrow U$ such that $\varphi(x)=x'$, $\varphi\circ\gamma'=\gamma$ on $I$, and along $\gamma[I]\cap U$, $\varphi_*g'=g$ and $\varphi_*\nabla'=\nabla$, where $\nabla$, $\nabla'$ are the Levi-Civita connections for $g$, $g'$ respectively.
\end{thm}
We note that, in the case of time-oriented manifolds, the local diffeomorphism in Theorem \ref{thm:FletcherWeatherall} can always be chosen so as to be orientation-preserving (along $\gamma$). We then have the following result:
\begin{thm}\label{thm:SEPforhadamard}
    Let $\mathscr{E}:\mathsf{Loc}\rightarrow \mathsf{FB}$ be a natural vector bundle, $\mathscr{P}:\mathscr{E}\rightarrow\mathscr{E}$ a natural real formally-Hermitian Green-hyperbolic operator, and let $\Pi$ be an assignment (not necessarily functorial), to every object $(M,g)$ in $\mathsf{Loc}$, of a distributional section of $\mathscr{E}^*(M,g)\boxtimes\mathscr{E}^*(M,g)$. Further suppose that $\mathrm{WF}(\Pi)$ satisfies the equivalence principle, in the sense discussed above. Let $(M,g)$ be Minkowski spacetime. Then if $\mathrm{WF}(\Pi(M,g))=\{(x,k;x',-k')\in\dot{N}M\times\dot{N}M:(x,k)\sim(x',k')\;\mathrm{and}\; k\in \dot{N}^+M\}$, where $\sim$ means `there exists a geodesic from $x$ to $x'$ with cotangent vector $k$ at $x$ and $k'$ at $x'$', the same is true, at least locally, for any other globally hyperbolic, time-oriented $(M',g')$.

    Moreover, the same is true if we replace $\dot{N}$ with $\dot{J}$ and $\dot{N}^+$ with $\dot{J}^+$, and if $\Pi$ satisfies (antisymmetric part) and (positivity) at each spacetime, the result holds globally.
\end{thm}
\begin{proof}
    Let $p'\in M'$ and let $U'\subset M'$ be (an open subset of) a geodesically normal neighbourhood of $p'$ (such always exist locally). We begin with the `$\supseteq$' direction. Let $(x,k;x',-k')\in\dot{N}M'\times\dot{N}M'$ restricted to $U'$ be such that $(x,k)\sim(x',k')$ and $k\in \dot{N}^+M'$. Let $\gamma':I'\rightarrow U'$ be a smooth embedded curve witnessing the relationship $(x,k)\sim(x',k')$. Fix any $p\in M$. Theorem \ref{thm:FletcherWeatherall} implies that there exists a neighbourhood $U$ of $p$, a smooth embedded curve $\gamma:I\rightarrow M$ whose image contains $p$ and a local diffeomorphism $\varphi:U'\rightarrow U$ such that $\varphi(x)=p$, $\varphi\circ\gamma'=\gamma$ on $I$, and along $\gamma[I]$, $\varphi_*g'=g$, $\varphi_*\nabla'=\nabla$, and which preserved time orientation there (that this is in fact the case for the whole of $\gamma[I]$ follows since $U'$ is a geodesically normal neighbourhood, see \parencite{Iliev2006}). Since $\mathrm{WF}(\Pi)$ satisfies the equivalence principle, we know that $\varphi_*(\mathrm{WF}(\Pi(M',g'))\cap\pi^{-1}_{T^*M'}(\gamma'[I']\times\gamma'[I']))=\mathrm{WF}(\Pi(M,g))\cap\pi^{-1}_{T^*M}(\gamma[I]\times\gamma[I])$, since for any $y',z'\in \gamma'[I']$, $\varphi_*(\mathrm{WF}(\Pi(M',g'))(y',z'))=\mathrm{WF}(\Pi(M,g))(\varphi(y'),\varphi(z'))$. Moreover, since along $\gamma[I]$, $\varphi_*g'=g$ and $\varphi_*\nabla'=\nabla$, it follows that $(\varphi(x),\varphi_*k;\varphi(x'),-\varphi_*k')\in\dot{N}M\times\dot{N}M$, $(\varphi(x),\varphi_*k)\sim(\varphi(x'),\varphi_*k')$, and $\varphi_*k\in \dot{N}^+M$. So $(\varphi(x),\varphi_*k;\varphi(x'),-\varphi_*k')\in\mathrm{WF}(\Pi(M,g))$ and hence $(x,k;x',-k')\in\mathrm{WF}(\Pi(M',g'))$.

    For the `$\subseteq$' direction, suppose that $(x,k;x',-k')\in \mathrm{WF}(\Pi(M',g'))$ restricted to $U'$. If $(x,k;x',-k')\in\dot{N}M'\times\dot{N}M'$, $(x,k)\sim(x',k')$, and $k\in \dot{N}^+M'$ then we are done, so suppose otherwise. Since $U'$ is geodesically normal, there exists a geodesic $\gamma':I'\rightarrow U'$ whose image contains $x$ and $x'$. Let $p$, $\gamma$, and $\varphi$ be defined subject to the same conditions as before. Again using that $\varphi\circ\gamma'=\gamma$ on $I$, and $\varphi_*g'=g$, $\varphi_*\nabla'=\nabla$ along $\gamma[I]$, we must have  $(\varphi(x),\varphi_*k;\varphi(x'),-\varphi_*k')\notin\dot{N}M\times\dot{N}M$ or  $(\varphi(x),\varphi_*k)$ and $(\varphi(x'),\varphi_*k')$  are not null related, or $\varphi_*k\notin \dot{N}^+M$. But since $\mathrm{WF}(\Pi)$ satisfies the equivalence principle, we know that $\varphi_*(\mathrm{WF}(\Pi(M',g'))\cap\pi^{-1}_{T^*M'}(\gamma'[I']\times\gamma'[I']))=\mathrm{WF}(\Pi(M,g))\cap\pi^{-1}_{T^*M}(\gamma[I]\times\gamma[I])$, and so $(\varphi(x),\varphi_*k;\varphi(x'),-\varphi_*k')\in\mathrm{WF}(\Pi(M,g))$, contradicting the assumption that $\mathrm{WF}(\Pi(M,g))=\{(x,k;x',-k')\in\dot{N}M\times\dot{N}M:(x,k)\sim(x',k')\;\mathrm{and}\; k\in \dot{N}^+M\}$.

    Moreover, the same would clearly follow if we replace $\dot{N}$ with $\dot{J}$ and $\dot{N}^+$ with $\dot{J}^+$ throughout. The remainder of the result follows from the local-to-global theorem of \textcite{Radzikowski1996}.
\end{proof}
In other words, given any assignment to every (globally hyperbolic, time-oriented) spacetime of a candidate two-point function whose wavefront set satisfies the equivalence principle, if its wavefront set agrees with that of the Minkowski vacuum of the Klein--Gordon field at Minkowski spacetime, the same holds, at least locally, at every other (globally hyperbolic, time-oriented) spacetime. As such, Theorem \ref{thm:SEPforhadamard} provides, for the first time, a precise way of making sense of the idea that the Hadamard condition can be motivated via appeal to the equivalence principle, applied to the singular structure of the Minkowski vacuum two-point function, i.e., argument (ii) from \S\ref{sec:intro}. 

\section{Discussion}\label{sec:conclusion}

Theorems \ref{thm:operatorstohadamard} and \ref{thm:SEPforhadamard} substantially clarify and generalise both the standard arguments for the Hadamard condition (i) and (ii) discussed in \S\ref{sec:intro}. They also open up several new directions for future research, especially regarding the relationship between the equivalence principle and Hadamard condition, and more generally the role of the equivalence principle in quantum field theory on curved spacetime. To end, I sketch some of these directions.

As noted in \S\ref{sec:equivalenceagain}, the definition of the equivalence principle in \parencite{March2025BPhil,MarchWeatherallFun1} cannot be applied directly to distributional quantities. This is because the equivalence principle concerns a property of certain functorial objects (fibre bundles, geometric object fields, etc.) which holds at a point---\textit{viz.}\ agreement at that point whenever two metrics agree to first order there---and there is no unambiguous definition of what it is for two distributions to `agree at a point' (unless those distributions are regular).\footnote{In particular, note that the pullback of a distribution $u$ on $M$ to the $0$-dimensional manifold along some (smooth) embedding is well-defined iff the image of that embedding lies outside of the singular support of $u$, see theorem 8.2.4 of \parencite{Hoermander1}.} However, many important objects in the theory of quantum fields on curved spacetime, such as the propagators of a real formally-Hermitian Green-hyperbolic operator, are distributional. As such, it would be of interest to provide a statement of the equivalence principle appropriate for distributional sections, e.g., perhaps as measured by some family of test functions, where the degree of approximation of the equivalence principle improves as the support of those test functions decreases, or for distributions considered as an appropriate scaling limit of some family of (smooth) sections.

Given a statement of the equivalence principle appropriate for distributional sections, one could then seek to prove theorems to the effect of, if $P$ is an equivalence principle-satisfying real formally-Hermitian Green-hyperbolic operator, then e.g., its causal propagator will satisfy the equivalence principle; if it is $\mathcal{V}^\pm$-decomposable, it will be so for such-and-such possible choices of $\mathcal{V}^\pm$ (presumably, $\mathcal{V}^\pm=\dot{N}^\pm$ or $\mathcal{V}^\pm=\dot{J}^\pm$); its associated space of Hadamard two-point functions at each spacetime will also satisfy the equivalence principle (in a sense appropriate for spaces of distributions), and so on. Doing so would provide control over the way in which the order of dependence on the metric in the dynamics of a quantum field theory on curved spacetime constrains other important structural features of that quantum field theory. It would also be a further contribution to the project of articulating clearly the relationship between the equivalence principle and the Hadamard condition. The advantage of theorem \ref{thm:SEPforhadamard} is to highlight exactly where this technical work is needed, and what we might expect it to achieve.  

\section*{Acknowledgments}
This material is based upon work supported by the National Science Foundation under Grant No.\ 2419967. I would like to thank Jeff Barrett, India Bhalla-Ladd, Silvester Borsboom, Chris Fewster, Frank Hu, James Read, Dominic Ryder, Jim Weatherall, and audiences in Lake Bohinj and Irvine for helpful comments and discussions. 

\appendix

\section{Extension of theorem \ref{thm:operatorstohadamard} to fermionic quantisation}\label{app:fermions}

Here, I briefly discuss the changes needed to extend the quantization procedure from \S\ref{sec:operators} to fermionic fields, and resulting modifications to theorem \ref{thm:operatorstohadamard}. References for this section include \parencite{Rejzner2011,Fewster2025Hadamard,Beninietal2015}. In fermionic field theories, we assume that $E$ is additionally equipped with a smooth Hermitian bundle automorphism $R:E\rightarrow E$ such that $RC=-CR$. The first modification is to the space of observables. A \emph{fermionic polynomial observable} is a functional $A:\bigoplus_{i=0}\otimes^i\Gamma_c(E)\rightarrow \mathbb{C}$ such that $A(f)=\sum_{i=0}\langle\alpha^{(i)},f^i\rangle$ for compactly supported distributional sections $\alpha^{(i)}$ of $\boxtimes^i_\mathrm{as}E$, where $f=\bigoplus_{i=0}f^i$ and the subscript $\mathrm{as}$ denotes antisymmetrization. (Observe that a bosonic polynomial observable as defined in \S\ref{sec:observables} is, too, a functional $A:\bigoplus_{i=0}\otimes^i\Gamma_c(E)\rightarrow \mathbb{C}$.) We denote the space of fermionic polynomial observables on $E$ as $\mathcal{P}^f(E)$. A fermionic polynomial observable $A^i:\bigoplus_{i=0}\otimes^i\Gamma_c(E)\rightarrow \mathbb{C}$ such that $A^i(f)=\langle\alpha^{(i)},f^i\rangle$ for some $i\in\mathbb{N}$, we say that $A^i$ is homogeneous of degree $i$. 
The definitions of regular, $\mathcal{V}$-microcausal fermionic polynomial observables are as before; likewise for on-shell observables.

The Poisson--Peierls bracket for a real formally-Hermitian Green-hyperbolic operator $P$ on $\mathcal{P}^f_{\mathcal{V}\mathrm{mc}}(E)$ is defined via
\begin{equation}
    [A_1,A_2]_P:=\mathrm{prod}_\mathrm{as}\circ \left\langle \Delta_P\circ(\mathrm{id}\otimes R), \mathrm{d}_L\otimes\mathrm{d}_R\right\rangle A_1\otimes A_2,\label{eq:Peierlsfermions}
\end{equation}
where $\mathrm{prod}_\mathrm{as}$ is the antisymmetrized multiplication operator, defined by $\mathrm{prod}_\mathrm{as}\circ(A^i_1\otimes A^j_2)(f^i\otimes f^j)=(\alpha^{(i)}_1\wedge\alpha_2^{(j)})(f^i\otimes f^j)$ with $\wedge$ the exterior product, for homogeneous $A^i,A^j\in \mathcal{P}^f_{\mathcal{V}\mathrm{mc}}(E)$ and extended by linearity in each argument, and $\mathrm{d}_L$ (respectively, $\mathrm{d}_R$) denotes the antisymmetrized functional derivative acting from the left (respectively, right), defined via $(\mathrm{d}_L A^i(f^{i-1}))(g^1)=\langle\alpha^{(i)},g\otimes f^{i-1}\rangle$ (respectively, $(\mathrm{d}_R A^i(f^{i-1}))(g^1)=\langle\alpha^{(i)}, f^{i-1}\otimes g\rangle$) for homogeneous $A^i\in \mathcal{P}^f_{\mathcal{V}\mathrm{mc}}(E)$ and again extended by linearity. Observe that if $A_1,A_2$ are homogeneous of odd degrees, $[A_1,A_2]_P$ is an anticommutator. 

In fermionic quantization, we again seek a star product on $\mathcal{P}^f_{\mathcal{V}\mathrm{mc}}(E)$ subject to conditions (1)--(3) from \S\ref{sec:operators}. These are fulfilled under a slight modification to the conditions discussed in \S\ref{sec:operators}, namely, (3) is now the case iff $\Pi-{}^t\Pi=i\Delta_P\circ(\mathrm{id}\otimes R)$. Correspondingly, we have a modified definition of a (fermionic) $\mathcal{V}^+$-Hadamard two-point function $\Delta^f_H$, in which (antisymmetric part) is replaced with $\Delta^f_H-{}^t\Delta^f_H=\Delta^f_H-\overline{\Delta}{}^f_H= i\Delta_P\circ(\mathrm{id}\otimes R)$. Using that $\mathrm{WF}(\Delta_P\circ(\mathrm{id}\otimes R))=\mathrm{WF}(\Delta_P)$, it is straightforward to see that the proof of Theorem \ref{thm:operatorstohadamard} carries over the the fermionic case under these modifications. 

\printbibliography

\end{document}